\documentclass[12pt]{amsart}
\usepackage{txfonts}      
\usepackage{amssymb}
\usepackage{eucal}
\usepackage{amsmath}
\usepackage{amscd}
\usepackage[dvips]{color}
\usepackage{multicol}
\usepackage[all]{xy}           
\usepackage{graphicx}
\usepackage{color}
\usepackage{colordvi}
\usepackage{xspace}
\usepackage{tikz}
\usepackage{makecell}

\usepackage{ifpdf}
\ifpdf
\usepackage[colorlinks,final,backref=page,hyperindex]{hyperref}
\else
\usepackage[colorlinks,final,backref=page,hyperindex,hypertex]{hyperref}
\fi

\usepackage[active]{srcltx} 

\usepackage{bbding}
\usepackage[shortlabels]{enumitem}
\usepackage{tikz-cd}

\makeatletter

\newtheorem{thm}{Theorem}[section]

\newtheorem{cor}[thm]{Corollary}
\newtheorem{pro}[thm]{Proposition}
\theoremstyle{definition}   
\newtheorem{defi}[thm]{Definition}

\newtheorem{ex}[thm]{Example}
\newtheorem{rmk}[thm]{Remark}

\newcommand{\nc}{\newcommand}
\newcommand{\delete}[1]{}

\nc{\mlabel}[1]{\label{#1}}  
\nc{\mcite}[1]{\cite{#1}}  
\nc{\mref}[1]{\ref{#1}}  
\nc{\meqref}[1]{\eqref{#1}}  
\nc{\mbibitem}[1]{\bibitem{#1}} 

\delete{
\nc{\mlabel}[1]{\label{#1}{\hfill \hspace{1cm}{\tt{{\ }\hfill(#1)}}}}
\nc{\mcite}[1]{\cite{#1}{{\tt{{\ }(#1)}}}}  
\nc{\mref}[1]{\ref{#1}{{\tt{{\ }(#1)}}}}  
\nc{\meqref}[1]{\eqref{#1}{{\tt{{\ }(#1)}}}}  
\nc{\mbibitem}[1]{\bibitem[\tt #1]{#1}} 
}

\newcommand {\emptycomment}[1]{}

\newcommand{\bk}{{\mathbf{k}}}
\nc{\vep}{\varepsilon}

\nc{\oprn}{\theta}
\nc{\Oprn}{\Theta}

\nc{\tforall}{\ \ \text{for all }}

\nc{\calo}{\mathcal{O}}
\nc{\oop}{$\mathcal{O}$-operator\xspace}
\nc{\oops}{$\mathcal{O}$-operators\xspace}
\nc{\mrho}{{\bm{\varrho}}}
\nc{\emk}{\mathbf{K}}
\nc{\invlim}{\displaystyle{\lim_{\longleftarrow}}\,}
\nc{\ot}{\otimes}

\newcommand{\lon }{\,\rightarrow\,}
\newcommand{\be }{\begin{equation}}
\newcommand{\ee }{\end{equation}}

\newcommand{\gr}{\mathsf{gr}}

\newcommand{\g}{\mathfrak g}

\newcommand{\G}{\mathbb G}

\newcommand{\huaR}{\mathcal{R}}

\newcommand{\huaF}{\mathcal{F}}

\nc{\calp}{\mathcal{P}}

\newcommand{\frkR}{\mathfrak R}

\newcommand{\br}[1]{   [ \cdot,    \cdot  ]   }

\newcommand{\BCH}{\mathsf{BCH}}

\nc{\NBS}{\mathsf{NBS}}

\newcommand{\fVec}{\mathsf{fVect}}

\newcommand{\cVec}{\mathsf{cVect}}

\nc{\CV}{\mathbf{C}}

\begin{document}
\title{
    Formal integration of complete Rota-Baxter  Lie algebras and Magnus expansion
}

\author{Maxim Goncharov}
\address{Sobolev Institute of Mathematics, Acad. Koptyug ave. 4, Novosibirsk, Russia}
\email{goncharov.gme@math.nsc.ru.}

\author{Pavel Kolesnikov}
\address{Sobolev Institute of Mathematics, Acad. Koptyug ave. 4, Novosibirsk, Russia}
\email{pavelsk77@gmail.com}

\author{Yunhe Sheng}
\address{Department of Mathematics, Jilin University, Changchun 130012, Jilin, China}
\email{shengyh@jlu.edu.cn}

\author{Rong Tang}
\address{Department of Mathematics, Jilin University, Changchun 130012, Jilin, China}
\email{tangrong@jlu.edu.cn}

\begin{abstract}
In this paper, first we revisit the  formal integration of   Lie algebras, which give rise to braces in some special cases. Then we establish  the formal integration theory for complete Rota-Baxter Lie algebras, that is, we show that there is a Rota-Baxter group with the underlying group structure given by the Baker-Campbell-Hausdorff  formula, associated to any complete Rota-Baxter Lie algebra. In particular, we use the post-Lie Magnus expansion to give the explicit formula of the Rota-Baxter operator. Finally we show that one can obtain a graded Rota-Baxter Lie ring from a filtered Rota-Baxter group.

\end{abstract}

\keywords{Rota-Baxter Lie algebra, integration, Rota-Baxter group, brace, Magnus expansion}

\renewcommand{\thefootnote}{}
\footnotetext{2020 Mathematics Subject Classification. 22E60, 
16T25, 
17B38, 
16T05, 
}

\maketitle

\tableofcontents
\section{Introduction}

The notion of Rota-Baxter operators on associative algebras was introduced by G. Baxter \cite{Bax} and they are applied in the
Connes-Kreimer's algebraic approach to renormalization of quantum
field theory~\cite{CK}, noncommutative
symmetric functions and noncommutative Bohnenblust-Spitzer identities \cite{Fard}, splitting of operads \cite{BBGN}, double Lie algebras \cite{GK},  etc. (see the book \cite{Gub} for more details).
In the Lie algebra context,  Rota-Baxter operators of weight 0 lead to the classical Yang-Baxter equation and integrable systems \cite{Bai,Ku}, Rota-Baxter operators of weight 1  are in one-to-one correspondence  with
solutions of the modified Yang-Baxter equation \cite{STS} and give rise to factorizations of Lie algebras. Moreover, from a more algebraic approach, a  Rota-Baxter operator naturally gives rise to a pre-Lie algebra or a post-Lie
algebra, and plays an important role in mathematical
physics~\cite{Aguiar,BGN}.

To study the global symmetry of Rota-Baxter Lie algebras, the notion of Rota-Baxter operators on groups was introduced in \cite{GLS}, and further studied in \cite{BG}. One can obtain Rota-Baxter operators of weight 1 on Lie algebras from that on Lie groups by differentiation. Later in \cite{JSZ}, it was shown that every Rota-Baxter Lie algebra of weight 1 can be integrated to a local
Rota-Baxter Lie group, namely the Rota-Baxter operator was only defined in an open neighborhood of the identity in the Lie group. It remains to be an open problem whether every Rota-Baxter Lie algebra of weight 1 can be integrated to a Rota-Baxter Lie group.
 On the other hand, the notion of  Rota-Baxter operators on cocommutative Hopf algebras  was given in  \cite{Goncharov} such that many classical results still hold in the Hopf algebra level. Then the notions of post-groups and post-Hopf algebras were introduced in \cite{BGST} and \cite{LST}, which are the split algebraic   structures underlying Rota-Baxter operators on groups and Hopf algebras respectively. In particular, a commutative post-group is called a pre-group, which is equivalent to a brace. The notion of a brace was introduced by Rump   in \cite{Rump1} to understand  the structure of set-theoretical solutions of the Yang-Baxter equation, which was proposed by Drinfel'd \cite{Dr}.

 In this paper, we use the formal integration approach to study integration of Rota-Baxter Lie algebras. For this purpose, we first revisit formal integration of Lie algebras, and characterize group-like elements in the completion of the universal enveloping algebra of a filtered Lie algebra using the exponential map explicitly. For a complete Lie algebra $(\g, \huaF_{\bullet}\g)$, the formal integration $(\g,*)$ is a group, where the group structure $*$ is given by the Baker-Campbell-Hausdorff  formula. As applications, we show that formal integration of certain nilpotent Lie algebras give rise to braces naturally, which lead to possible applications in the Yang-Baxter equation and pre-Lie algebras \cite{Rump1,Sm22b}. Then we show that the universal enveloping algebra of a filtered Rota-Baxter Lie algebra is a filtered Rota-Baxter Hopf algebra. Finally, by completing the filtered Rota-Baxter universal enveloping algebra and considering the set of group-like elements, we obtain a Rota-Baxter group. We summarize it in the following diagram:\vspace{-2mm}
 \begin{equation*}
    \begin{split}
        \xymatrix{
            \text{filtered RB Lie}\atop (\g,\huaF_{\bullet}\g,R) \ar[r]^{U(\cdot)}     &  \text{filtered RB Hopf}\atop (U(\g),\huaR)\ar[r]^{\widehat{\cdot}} &
            \text{Complete RB Hopf}\atop (\widehat U(\g),\widehat \huaR) \ar[r]^{\quad G(\cdot)}  & \text{RB group}\atop (G,\cdot,\widehat \huaR)  \ar@<.3ex>[r]^{\log} &
        \text{RB group} \atop (\widehat\g,*,\frkR)\ar@<.3ex>[l]^{\exp}}
    \end{split}
\end{equation*}
 In particular, for a complete Rota-Baxter Lie algebra $(\g, \huaF_{\bullet}\g, R)$, we obtain a Rota-Baxter group  $(\g,*,\frkR)$, which serves as the formal integration. It is worth noting that the post-Lie Magnus expansion \mcite{CP,CEO,EMQ,EMM,MQS} naturally appears in the explicit expression of the Rota-Baxter operator $\frkR$ on the group $ (\widehat{\g},*)$. Conversely, we show that associated to a filtered Rota-Baxter group, one can naturally obtain a graded Rota-Baxter Lie ring.

 The paper is organized as follows. In Section \ref{sec:pre}, we recall the notions of filtered vector spaces, filtered Lie algebras, filtered Hopf algebras and their relations. In Section \ref{sec:intLie}, we revisit formal integration of Lie algebras, and establish the connection with braces. In Section \ref{sec:intRB}, first we show that the universal enveloping algebra of a filtered Rota-Baxter Lie algebra is a filtered Rota-Baxter Hopf algebra. Then using the trick of completion, we associate to every filtered Rota-Baxter Lie algebra a Rota-Baxter group. In Section \ref{sec:diffRB}, we show that one can obtain a graded Rota-Baxter Lie ring from
 a filtered Rota-Baxter group.


Throughout the paper, $\bk$ is an arbitrary field of characteristic zero.\vspace{2mm}

\noindent
{\bf Acknowledgments.} We thank Prof. Efim Zelmanov for very helpful comments that lead us to obtain the results in the last part.
The first and the second authors were supported by RAS Fundamental Research Program (FWNF-2022-0002), the third author is supported by
NSFC (12471060) and the fourth author is supported by
NSFC (12371029) and the Fundamental Research Funds for the Central Universities.

\section{Preliminaries: filtered vector spaces, Lie algebras, and Hopf algebras.} \label{sec:pre}

In this section, we recall the notions of filtered vector spaces, filtered Lie algebras, filtered Hopf algebras and the relations between filtered Lie algebras and filtered Hopf algebras.

\begin{defi}\mcite{Fr}
A {\bf   filtered vector space} is a pair $(V,\huaF_{\bullet}V)$, where $V$ is a vector space and $\huaF_{\bullet}V$ is a descending
filtration of the  vector space $V$ such that $$V=\huaF_0V\supset\huaF_1V\supset\cdots\supset\huaF_n V\supset\cdots.$$

Let $(V,\huaF_{\bullet}V)$ and $(W,\huaF_{\bullet}W)$ be filtered  vector spaces. A {\bf homomorphism} $f:(V,\huaF_{\bullet}V)\to (W,\huaF_{\bullet}W)$ is a linear map $f:V\to W$ such that $f(\huaF_n V)\subset \huaF_n W,~n\geq 0.$
\end{defi}

Denote by  $\fVec$ the category of filtered vector spaces.
For filtered vector spaces $(V,\huaF_{\bullet}V)$ and $(W,\huaF_{\bullet}W)$, there is a filtration on $V\otimes W$   given by
\begin{eqnarray}\mlabel{FTP}
\huaF_n(V\otimes W)\coloneqq \sum_{i+j=n}\huaF_iV\otimes\huaF_jW\subset V\otimes W.
\end{eqnarray}
Thus, $(V\otimes W,\huaF_{\bullet}(V\otimes W))$ is a filtered vector space, called the  {\bf tensor product} of filtered vector spaces  $(V,\huaF_{\bullet}V)$ and $(W,\huaF_{\bullet}W)$. For the ground field $\bk$, we define $\huaF_0\bk=\bk,\huaF_n\bk=\{0\},~n\geq 1$. It follows that $(\fVec,\otimes,\bk)$ is a $\bk$-linear  symmetric monoidal category.

\begin{pro}Let $V$ be a vector space equipped with two filtrations $\huaF_{\bullet}V$ and $\huaF'_{\bullet}V$. Define the intersection of filtrations $(\huaF\cap \huaF')_{\bullet}V$ as
$$
(\huaF\cap \huaF')_nV=\huaF_nV\cap \huaF'_nV.
$$
Then $(V, (\huaF\cap \huaF')_{\bullet}V)$ is also a filtered vector space.
\end{pro}

\begin{defi}\cite{Quillen}
   A {\bf filtered Lie algebra} is a pair $(\g, \huaF_{\bullet}\g)$, where $\g$ is a Lie algebra and $\huaF_{\bullet}\g$ is a descending filtration of the vector space $\g$ such that $\g=\huaF_1\g\supset \huaF_2\g\supset\ldots\supset \huaF_n\g\supset\ldots$ and
   \begin{equation}\label{g1}
   [\huaF_n\g,\huaF_m\g]\subset \huaF_{m+n}\g
   \end{equation}
   for all $m,n\geq 1$.
\end{defi}
Note that a filtration here starts from degree 1,
so filtered Lie algebras are not objects in $\fVec$.

Obviously, if $(\g,\huaF_{\bullet}\g)$ is a filtered Lie algebra, then for any $n\geq 1$, $\huaF_n\g$ is an ideal in $\g$. Now we give two natural examples of filtration structures on a Lie algebra.

\begin{ex}\label{trivial-fil}
For an arbitrary Lie algebra $\g$, one can define a trivial filtration by
$$
\huaF_n\g=\g,\quad n\geq 1.
$$
    \end{ex}

\begin{ex}\label{filtrst}
 By the standard filtration on a Lie algebra $\g$, we will mean the following filtration $\huaF_{\bullet}\g$ on $\g$
$$
\huaF_n\g=\g^n,\quad n\geq 1,
$$
where $\g^n=[\g,\g^{n-1}]$ for $n>1$ and $\g^1=\g$.
\end{ex}

\begin{pro}
    Let ($\g,\huaF_{\bullet}\g)$ and $(\g,\huaF'_{\bullet}\g)$ be two  filtered Lie algebra structures on the same Lie algebra $\g$. Then $(\g,(\huaF\cap \huaF')_{\bullet}\g)$ is also a filtered Lie algebra.
\end{pro}

Note that \eqref{g1} means that the product $[\cdot,\cdot]:\g\otimes \g\longrightarrow \g$ is a homomorphism of filtered vector spaces.

For a Hopf algebra $H=(H,\mu, \Delta, \eta, \epsilon, S)$, we  will use the following notations:
    \begin{itemize}
    \item $\mu: H\otimes H \longrightarrow H$ is a multiplication,
     \item $\Delta: H\longrightarrow H\otimes H$ is a comultiplication,
     \item $\eta: \bk\longrightarrow H$ is a unit,
     \item $\epsilon: H\longrightarrow \bk$ is a counit,
     \item $S: H\longrightarrow H$ is the antipode.
\end{itemize}

Given a coalgebra $(A,\Delta)$, we will use the Sweedler's notation: for any $a\in A$
$$
\Delta(a)= a_{(1)}\otimes a_{(2)},
\quad
(\Delta\otimes \mathrm{id})\Delta(a)= a_{(1)}\otimes a_{(2)}\otimes a_{(3)},
\quad \text{etc.}
$$

Given a Hopf algebra $H$, by $P(H)$ we will denote the set of primitive elements of $H$
$$
P(H)=\{x\in H|\ \Delta(x)=x\otimes 1+1\otimes x\}.
$$

\begin{defi}
A {\bf filtered Hopf algebra} is a Hopf algebraic object in the category $(\fVec,\otimes, \bk)$. More precisely, it is a Hopf algebra $(H, \mu, \Delta, \eta, \epsilon, S)$ endowed with a filtration $\huaF_{\bullet}H$  such that:
\begin{itemize}
\item[(i)] all maps $(\mu, \Delta, \eta, \epsilon, S)$ are homomorphisms of the corresponding filtered vector spaces.
\item[(ii)] $\ker(\epsilon)=\huaF_1H$.
\end{itemize}
\end{defi}

Note that the second condition $\ker(\epsilon)=\huaF_1H$ is equivalent to the condition that $H/\huaF_1H\cong\bk$.

The following conclusion is straightforward. 

\begin{pro}\label{pro:Lie-Hopf}\cite{Fr,Quillen}  Let $(\g,\huaF_{\bullet}\g)$ be a filtered Lie algebra. Then the universal enveloping algebra $U(\g)$ is a filtered Hopf algebra with the filtration $\huaF_{\bullet}U(\g)$ defined as
\begin{gather}
\label{g2}\huaF_0U(\g)=U(\g),\\
\label{g3}\huaF_nU(\g)=span \left\{x_1\ldots x_k|\ k\geq 1,\
x_i\in \huaF_{n_i}\g,\ \sum_{i=1}^k n_i\ge n\ \right\},\quad\ n\geq 1.
\end{gather}
The filtered Hopf algebra $(U(\g),\huaF_{\bullet}U(\g))$  is called  the {\bf filtered universal enveloping algebra} of the filtered Lie algebra $(\g,\huaF_{\bullet}\g)$.
\end{pro}

  Consider the universal enveloping algebra of the filtered Lie algebras given in Example \ref{trivial-fil} and Example \ref{filtrst}, we have the following examples.

\begin{ex}\label{expow-1}
  The  filtered universal enveloping algebra $(U(\g),\huaF_{\bullet}U(\g))$ of the trivial  filtered Lie algebra $(\g,\huaF_{\bullet}\g)$ given in Example \ref{trivial-fil} is given by
\begin{eqnarray*}
\huaF_0U(\g)&=&U(\g),\\
\huaF_nU(\g)&=&\ker(\epsilon), \quad\ n\geq 1.
\end{eqnarray*}
\end{ex}

\begin{ex}\label{expow}
The  filtered universal enveloping algebra $(U(\g),\huaF_{\bullet}U(\g))$ of the standard filtered Lie algebra $(\g,\huaF_{\bullet}\g)$ given in Example \ref{filtrst} is given by
\begin{eqnarray*}
\huaF_0U(\g)&=&U(\g),\\
\huaF_nU(\g)&=&span\{x_1\ldots x_k|\ k\geq n,\ x_i\in \g\}, \quad n\geq 1.
\end{eqnarray*}

We call this filtration the standard filtration on $U(\g)$.
\end{ex}

\begin{pro}\label{indfilt}
    Let $(H, \mu, \Delta, \eta, \epsilon, S,\huaF_{\bullet}H)$ be a filtered Hopf algebra. Then the set $P(H)$ of primitive elements is a filtered Lie algebra with the induced filtration $$\huaF_nP(H)=\huaF_nH\cap P(H),\quad n\geq 1.$$
\end{pro}
\begin{proof}
  Since $P(H)\subset\ker(\epsilon)=\huaF_1 H$, we have that
  $$\huaF_1 P(H)=\huaF_1 H\cap P(H)=P(H).$$

For any $n\geq 1$, we have
    $$\huaF_nP(H)=\huaF_nH\cap P(H)     \supset \huaF_{n+1}H\cap P(H)=\huaF_{n+1}P(H).$$
Also, for any $i,j\geq 1$,
$$
[\huaF_iP(H),\huaF_jP(H)]\subset [\huaF_iH,\huaF_jH]\subset \huaF_{i+j}H.
$$
By $[P(H),P(H)]\subset P(H)$, we deduce  that $[\huaF_iP(H),\huaF_jP(H)]\subset \huaF_{i+j}P(H)$.
\end{proof}

\begin{rmk}
It follows from the Poincar\'e--Birkhoff--Witt Theorem that
the induced filtration on $\g$ from the standard filtration  $\huaF_{\bullet}U(\g)$   on the universal enveloping algebra $U(\g)$ of a Lie algebra $\g$  coincides with the standard filtration on $\g$.
\end{rmk}

\section{Formal integration of complete Lie algebras and braces}\label{sec:intLie}

In this section, we revisit formal integrations of complete Lie algebras, in which the completion of a filtered Hopf algebra is the main ingredient. In particular, we give the precise description of group-like elements in the completion of the universal enveloping algebra of a filtered Lie algebra. Moreover, we establish the connection between formal integration of complete Lie algebras and braces.

\begin{defi} \mcite{Fr}
A {\bf   complete vector space} is a filtered vector space $(V,\huaF_{\bullet}V)$ such that the natural homomorphism
\vspace{-.2cm}
\begin{eqnarray}\mlabel{complete-mor}
\Phi_V:V\to \underleftarrow{\lim}V/\huaF_n V
\end{eqnarray}
is a linear isomorphism of vector spaces. Denote by  $\cVec$ the category of complete vector spaces,
which is a full subcategory of $\fVec$.
\end{defi}

Let $(V,\huaF_{\bullet}V)$ be a filtered vector space. Then the vector space $\widehat{V}=\underleftarrow{\lim}V/\huaF_n (V)$ is a complete vector space with the filtration given by
\begin{eqnarray}
\huaF_n(\widehat{V})=\ker \pi_n,\,\,\pi_n:\widehat{V}\lon V/\huaF_n (V).
\end{eqnarray}
We call the complete vector space $\big(\widehat{V},\huaF_\bullet\widehat{V}\big)$ the {\bf completion} of the filtered vector space $(V,\huaF_{\bullet}V)$.
Let $(V,\huaF_{\bullet}V)$ and $(W,\huaF_{\bullet}W)$ be complete vector spaces. The completion of the filtered vector space $(V\otimes W,\huaF_{\bullet}(V\otimes W))$ is called the {\bf complete tensor product} of $(V,\huaF_{\bullet}V)$ and $(W,\huaF_{\bullet}W)$.
We denote the complete tensor product of $(V,\huaF_{\bullet}V)$ and $(W,\huaF_{\bullet}W)$ by $(V\hat{\otimes}W,\huaF_{\bullet}(V\hat{\otimes}W))$.

\begin{pro}{\rm \cite[Proposition 7.3.11]{Fr}}\label{iso-cv}
Let $(V,\huaF_{\bullet}V)$ and $(W,\huaF_{\bullet}W)$ be filtered vector spaces. Then the  complete vector space $V{\hat{\otimes}}W$ is isomorphic to the  complete vector space $\widehat{V}\hat{\otimes}\widehat{W}$.
\end{pro}

Moreover,  $(\cVec,\hat{\otimes},\bk)$ is a $\bk$-linear symmetric monoidal category. By Proposition \ref{iso-cv},  one can see that the completion functor $\widehat{(-)}$ is a  symmetric monoidal functor from the  category $\fVec$ to the  category $\cVec$. 

\begin{defi}
A filtered Lie algebra $(\g, \huaF_{\bullet}\g)$ is called {\bf complete} if there is an isomorphism of vector spaces
	$
	\g\cong\underleftarrow{\lim}~\g/\huaF_{n+1}\g
	$.
\end{defi}

Recall that
    a Lie algebra $\g$ is called nilpotent of nilindex $n\in \mathbb N$, if $\g^n=0$.

\begin{ex}\label{nilp}

Let $\g$ be a nilpotent Lie algebra of nilindex $n$.
Then the standard filtered Lie algebra $(\g,\huaF_{\bullet} \g)$ is a complete Lie algebra. The isomorphism between $\g$ and $\varprojlim \g/\huaF_{n+1}\g$  is given by
\begin{equation}\label{corresp}
x\mapsto (\pi_1(x),\pi_2(x), \ldots, \pi_{n-1}(x),x,x,\ldots,x,\ldots),
\end{equation}
where $\pi_i:\g\rightarrow \g/\huaF_{i+1}\g$ is the natural Lie algebra homomorphism.
\end{ex}

\begin{defi}
    A {\bf complete Hopf algebra} is a Hopf algebraic object in the category $(\cVec,\hat{\otimes}, \bk)$. More precisely, a complete Hopf algebra is a complete vector space $(H,\huaF_{\bullet}H)$ equipped with morphisms of complete vector spaces:
    \begin{itemize}
    \item $\mu: H\hat{\otimes} H\longrightarrow H$ is a multiplication,
     \item $\Delta: H\longrightarrow H\hat{\otimes} H$ is a comultiplication,
     \item $\eta: \bk\longrightarrow H$ is a unit,
     \item $\epsilon: H\longrightarrow \bk$ is a counit,
     \item $S: H\longrightarrow H$ is the antipode,
\end{itemize}
that satisfy the usual axioms of a Hopf algebra.
\end{defi}

Let $(\g,\huaF_{\bullet}\g)$ be a filtered Lie algebra and $(U(\g),\Delta,\eta,\epsilon, S,\huaF_{\bullet}U(\g))$ be the filtered universal enveloping algebra of $(\g,\huaF_{\bullet}\g)$.
Let $\pi_n: U(\mathfrak g)\rightarrow U(\g)/\huaF_{n+1}U(\mathfrak g)$ ($n\geq -1$) be the natural homomorphism. Let  $p_n: U(\mathfrak g)/\huaF_{n+1}U(\mathfrak g)\rightarrow U(\mathfrak g)/\huaF_nU(\mathfrak g)$ ($n\geq 1$) be the map defined by
$$
p_n(x+\huaF_{n+1}U(\mathfrak g))=x+\huaF_nU(\mathfrak g).
$$
Then $(U(\mathfrak g)/\huaF_nU(\mathfrak g), p_n)$ is an inverse system of algebras.

Let $\widehat{U}(\g)=\varprojlim (U(\mathfrak g)/\huaF_n(U(\mathfrak g))$  be the completion of the algebra $U(\mathfrak g)$. As an algebra, $\widehat{U}(\g)$ is a subalgebra in $\prod\limits_{n\geq 0} U(\mathfrak g)/\huaF_{n+1} U(\mathfrak g)$, the direct product of algebras $\left(U(\mathfrak g)/\huaF_{n+1} U(\mathfrak g)\right)_{n\geq 0}$ with component-wise operations. Moreover, as a vector space, $\widehat{U}(\g)$ consists of elements $$(x_0,\ldots,x_n,\ldots)\in \prod\limits_{n\geq 0} U(\mathfrak g)/\huaF_{n+1} U(\mathfrak g) $$ satisfying $p_n(x_n)=x_{n-1}$. Define a filtration on $\widehat{U}(\g)$ as
\begin{eqnarray}
\huaF_n(\widehat{U}(\g))=\ker \pi_{n-1},\ n\geq 0.
\end{eqnarray}

Consider the tensor product $\widehat{U}(\g)\otimes\widehat{U}(\g)$ as a filtered vector space with the filtration defined by \eqref{FTP}. Let $P_n:\widehat{U}(\g)\otimes\widehat{U}(\g)/\huaF_{n+1}(\widehat{U}(\g)\otimes\widehat{U}(\g))\rightarrow  \widehat{U}(\g)\otimes\widehat{U}(\g)/\huaF_{n}(\widehat{U}(\g)\otimes\widehat{U}(\g))$ be the natural homomorphism.   Then $(\widehat{U}(\g)\otimes\widehat{U}(\g), \huaF_{\bullet}(\widehat{U}(\g)\otimes\widehat{U}(\g)),P_n)$ is an inverse system of vector spaces. Define the complete tensor product $\hat{\otimes}$ as:
$$\widehat{U}(\g)\hat{\otimes}\widehat{U}(\g):=\varprojlim (\widehat{U}(\g)\otimes \widehat{U}(\g)/\huaF_{n+1}(\widehat{U}(\g)\otimes\widehat{U}(\g)))\cong\varprojlim ({U}(\g)\otimes {U}(\g)/\huaF_{n+1}({U}(\g)\otimes{U}(\g))).$$
Since {$\Delta: U(\g)\rightarrow U(\g)\otimes U(\g)$ is a homomorphism of filtered vector spaces}, for any $n\geq 0$,  we can define a   linear map  $\Delta_n: U(\mathfrak g)/\huaF_{n+1} U(\mathfrak g)\rightarrow U(\mathfrak g)\otimes U(\mathfrak g)/\huaF_{n+1}\left (U(\mathfrak g)\otimes U(\mathfrak g)\right)$
by
$$
\Delta_n(x+\huaF_{n+1}U(\mathfrak g))=\Delta(x)+\huaF_{n+1}\left (U(\mathfrak g)\otimes U(\mathfrak g)\right).
$$
It is straightforward to check that $\Delta_{n-1}\circ p_n=P_n\circ \Delta_n$. In this case, we can define a map $\widehat{\Delta}:\widehat{U}(\g)\rightarrow \widehat{U}(\g)\hat{\otimes}\widehat{U}(\g)$: for any $(x_0,\ldots,x_n,\ldots)\in \widehat{U}(\g)$ we put
$$
\widehat{\Delta}(x_0,\ldots,x_n,\ldots)=(\Delta_0(x_0),\ldots,\Delta_n(x_n),\ldots)\in \widehat{U}(\g)\hat{\otimes}\widehat{U}(\g).
$$

Similarly, we can define completions of the other Hopf algebra operations defined on $U(\g)$ and consider $\widehat{U}(\g)$ as a complete Hopf algebra.

Given $x=(0,x_1,\ldots,x_n,\ldots)\in \huaF_1\widehat{U}(\g)$, that is,
$x_1\in \huaF_1U(\g)$ and $x_{i+1}-x_i\in \huaF_{i+1}{U}(\g)$
for $i\ge 1$, we can define the exponential map $\exp:\huaF_1\widehat U(\g)\rightarrow \widehat U(\g)$ as
\begin{equation}\label{exp}
\exp(x)=\left (1,1+x_1,1+x_2+\frac{x_2^2}{2},\ldots, \sum\limits_{i=0}^n\frac{x_n^i}{i!},\ldots\right ):=1+x+\frac{x^2}{2}+\ldots+\frac{x^n}{n!}+\ldots.
\end{equation}

\begin{defi}
An element $x\in \widehat{U}(\g)$ is called a group-like element  if $\widehat{\Delta}(x)=x\hat{\otimes}x$ and $x\not=0$. We denote the set of the group-like elements of $\widehat{U}(\g)$ by $\G$.
\end{defi}

Since $\g$ is naturally embedded into $U(\g)$ and $\huaF_n \g\subset \huaF_n U(\g)$, we may consider the completion $\widehat{\g}=\varprojlim (\g/\huaF_{n+1}\g)$ as a subspace in $\huaF_1\widehat{U}(\g)$.

\begin{pro}\label{grexp}{\rm \cite[Proposition 2.6]{Quillen}}
The  map $\exp$
is a
bijective map from  $\widehat{\g}$ to $\G$.
\end{pro}

Now we have two bijections: $\exp:\widehat{\g}\rightarrow \G$ and $\log=\exp^{-1}:\G\rightarrow \widehat{\g}$. Consequently, the Baker-Campbell-Hausdorff formula gives rise to a group structure on the set $\widehat \g,$ which is the main ingredient in the study of formal integration of Rota-Baxter Lie algebras in the next section.

\begin{pro}\label{pro:group-BCH}
 Let $(\g,\huaF_{\bullet}\g)$ be a filtered Lie algebra and $\widehat \g$ be its completion. Then there is a group structure $*:\widehat\g\times\widehat\g\lon \widehat\g$ on $\widehat\g$ which is given by
\begin{eqnarray}\label{eq:BCH-group}
x* y\triangleq \BCH(x,y)= \log\big(\exp(x)\cdot\exp(y)\big)\tforall x,y\in\widehat\g.
\end{eqnarray}
Furthermore, both $\exp$ and $\log$ are group isomorphisms.
\end{pro}

\begin{rmk}
  If $(\g,\huaF_{\bullet}\g)$ is a complete   Lie algebra,
then
  the group $(\g,*)$   is viewed as its formal integration.
\end{rmk}

\begin{ex}\label{ex:Hei-group}
Let $\mathfrak g$ be the 3-dimensional Heisenberg Lie algebra with the basis $e_1,e_2,e_3$ satisfying
$$
[e_1,e_2]=e_3,\quad [e_3,e_1]=[e_3,e_2]=0.
$$
We have that $\g^3=0$ and the standard filtration $\huaF_1\g=\g$, $\huaF_2\g=span\{e_3\}$, $\huaF_3\g=0$.
So $(\g,\huaF_{\bullet}\g)$ is a complete  Lie algebra. Since $\g^3=0$, the group structure $*$ on $\g$ is given by
\begin{equation}\label{eq:HeiBCH-group}
x*y=x+y+\frac{1}{2}[x,y].
\end{equation}
\end{ex}

At the end of this section, we give an application of formal integration of complete Lie algebras, and establish the connection with braces.

\begin{defi}{\rm\cite{Rump1}}
  A {\bf brace} is a triple $(B, +, *)$, where
  $(B, +)$ is an abelian group and $(B, *)$ is a group such that the following compatibility condition holds:
 \begin{equation}\label{e9}
 x*(y+z)=(x*y)-x+(x*z),\quad \forall x,y,z\in B.
 \end{equation}
\end{defi}

\begin{thm}
Let $\mathfrak g$ be a nilpotent Lie algebra. If $\g^3=0,$ then $(\mathfrak g, +, *)$ is a brace, where $+$ is natural abelian group structure on the vector space $\g$ and  the group structure $*$ is given by \eqref{eq:BCH-group}.
\end{thm}

\begin{proof}
We need to check that for all $x,y,z\in \mathfrak g$, the equality \eqref{e9} holds. We have
\begin{eqnarray*}
  x*(y+z)&=&x+y+z+\frac{1}{2}[x,y+z]\\
  &=&(x+y+\frac{1}{2}[x,y])-x+(x+z+\frac{1}{2}[x,z])\\
  &=&(x*y)-x+(x*z).
\end{eqnarray*}
   The statement is proved.
\end{proof}

\begin{ex}
  Let $\mathfrak g$ be the 3-dimensional Heisenberg Lie algebra given in Example \ref{ex:Hei-group}. Then $(\g,+,*)$ is a brace, where $+$ is the natural abelian group structure on the vector space $\g$ and  the group structure $*$ is given by \eqref{eq:HeiBCH-group}.
\end{ex}

\section{Formal integration of a complete Rota-Baxter Lie algebra}\label{sec:intRB}

In this section first we show that the universal enveloping algebra of a filtered Rota-Baxter Lie algebra is a filtered Rota-Baxter Hopf algebra. Then by completing the Rota-Baxter operator in the filtered Rota-Baxter universal enveloping algebra, we obtain a Rota-Baxter group. In particular, for a complete Rota-Baxter Lie algebra $(\g, \huaF_{\bullet}\g, R)$, we obtain a Rota-Baxter group  $(\g,*,\frkR)$, which serves as the formal integration of the complete Rota-Baxter Lie algebra $(\g, \huaF_{\bullet}\g, R)$. An explicit formula for the Rota-Baxter operator $\frkR$ is given via the post-Lie Magnus expansion.

\begin{defi} \label{defi:O}
A {\bf Rota-Baxter Lie algebra} is a pair $((\g,[\cdot,\cdot]), R)$, where $(\g,[\cdot,\cdot])$ is a Lie algebra and $R:\g\longrightarrow\g$ is a  { Rota-Baxter operator of weight $1$}, i.e.
 \begin{equation}
[R(x),R(y)]=R\big([R(x),y]+ [x,R(y)]+[x,y] \big), \quad \forall x, y \in \g.
 \end{equation}
 \end{defi}
Define for any $x,y\in \g$,
$$
[x,y]_R=[R(x),y]+[x,R(y)]+[x,y].
$$
Then the space $\g$ with this  new product $[\cdot,\cdot]_R$ is a Lie algebra, and $R$ is a Lie algebra homomorphism from $(\g,[\cdot,\cdot]_R)$ to $(\g,[\cdot,\cdot])$. The Lie algebra $(\g,[\cdot,\cdot]_R)$ is called the {\bf descendant Lie algebra} of the Rota-Baxter Lie algebra $((\g,[\cdot,\cdot]), R)$. For simplicity, we denote the descendant Lie algebra by $\g_R$.

\begin{defi}{\rm \cite[Definition 2.1]{GLS}}
A {\bf Rota-Baxter group} is a pair $(G,\frkR)$, where $G$ is a group and $\frkR:G\rightarrow G$ is a map satisfying the following equality:
\begin{equation}\label{defi:1}
    \frkR(g)\frkR(h)=\frkR(g\frkR(g)h\frkR(g)^{-1}), \quad \forall g,h\in G.
\end{equation}
\end{defi}

Note that if $(G,\frkR)$ is a Rota-Baxter group and $x\in G$, then $\frkR(1)=1$ and $$\frkR(x)^{-1}=\frkR\Big(\frkR(x)^{-1}x^{-1}\frkR(x)\Big).$$ Define a new multiplication on $G$ by $$x\star y=x\frkR(x)y\frkR(x)^{-1},\quad \forall x,y\in G$$ Then $(G,\star)$ is a group, and $\frkR$ is a group homomorphism, that is, $\frkR(x)\frkR(y)=\frkR(x\star y)$.

\begin{defi} {\rm\cite{Goncharov}}
A {\bf Rota-Baxter Hopf algebra} is a pair $(H,\huaR)$, where $H$ is a cocommutative Hopf algebra and $\huaR:H\to H$ is a coalgebra homomorphism satisfying the following equality:
\begin{eqnarray}\label{RBH}
\huaR(h)\huaR(t)=\huaR\Big(h_{(1)
}\huaR(h_{(2)})tS(\huaR(h_{(3)}))\Big)
,\,\,\forall h,t\in H.\end{eqnarray}
\end{defi}

\begin{pro}\label{RBH-Hopf}
A Rota-Baxter Hopf algebra $(H,\huaR)$  induces a new Hopf algebra $(H,\star,\Delta, \eta, \epsilon, S_{\huaR})$, where the new product $\star$ and the antipode $S_{\huaR}$ are defined as
\begin{eqnarray}
x\star y&=&x_{(1)}\huaR(x_{(2)})yS(\huaR(x_{(3)})), \label{nprodU}\\
\nonumber  S_{\huaR}(x)&=&S(\huaR(x_{(1)}))S(x_{(2)})\huaR(x_{(3)}).
\end{eqnarray}
\end{pro}

Now we introduce the notions of filtered Rota-Baxter Lie algebras and filtered Rota-Baxter Hopf algebras.

 \begin{defi}
 A {\bf filtered Rota-Baxter Lie algebra} is a triple $(\g, \huaF_{\bullet}\g , R)$, where $(\g,\huaF_{\bullet}\g)$ is a filtered Lie algebra, $R$ is a Rota-Baxter operator of weight 1 on $\g$ and
    $$R(\huaF_n\g)\subset \huaF_n\g,\quad \forall n\geq 1.$$
   A  filtered Rota-Baxter Lie algebra $(\g, \huaF_{\bullet}\g , R)$ is called a {\bf complete Rota-Baxter Lie algebra}  if the underlying filtered Lie algebra $(\g, \huaF_{\bullet}\g)$ is complete.
\end{defi}

\begin{defi}
A {\bf filtered Rota-Baxter Hopf algebra} is a triple $(H,\huaF_{\bullet} H,\huaR)$, where $(H,\huaF_{\bullet} H)$ is a filtered Hopf algebra and  $\huaR$ is a Rota-Baxter operator on $H$ satisfying
$$
\huaR(\huaF_nH)\subset \huaF_nH,\quad \forall n\geq 0.
$$
\end{defi}

\begin{pro}
    Let $(H,\huaF_{\bullet} H,\huaR)$ be a filtered Rota-Baxter Hopf algebra. Then $(P(H), \huaF_{\bullet}H\cap P(H), \huaR|_{P(H)})$ is a filtered Rota-Baxter Lie algebra.
\end{pro}
\begin{proof}
   From \cite{Goncharov} it follows that $\huaR(P(H))\subset P(H)$ and $(P(H), \huaR|_{P(H)})$ is a Rota-Baxter Lie algebra.  Since $\huaR(\huaF_n H)\subset \huaF_n H$, we have $$\huaR(\huaF_n P(H))=\huaR(\huaF_n H\cap P(H))\subset \huaF_n H\cap P(H)=\huaF_n P(H),$$
   as required.
\end{proof}
The connection between Rota-Baxter operators of weight 1 on a Lie algebra $\g$ and Rota-Baxter operators on the universal enveloping algebra was obtained in \cite{Goncharov}. More precisely, if $\g$ is a Lie algebra and $R:\g\longrightarrow \g$ is a Rota-Baxter operator of weight 1, then the map $\huaR: U(\g)\longrightarrow U(\g)$ defined as
\begin{eqnarray}
    \huaR(1)&=&1 \label{e100},\\
    \huaR(x)&=&R(x),\ \ x\in \g, \label{e10}\\
    \huaR(xh)&=&R(x)\huaR(h)-\huaR([R(x),h]),\ \ x\in \g,
     \ h\in U(\g), \label{e11}
    \end{eqnarray}
is a well-defined coalgebra homomorphism and satisfies \eqref{RBH} (that is, $(U(\g),\huaR)$ is a Rota-Baxter Hopf algebra).

Now we generalize this result to the filtered setting.
\begin{thm}\label{thm:RBLie-RBH}
    Let $(\g, \huaF_{\bullet}\g, R)$ be a filtered Rota-Baxter Lie algebra and $(U(\g),\huaF_{\bullet}U(\g))$ be the filtered universal enveloping algebra given in Proposition \ref{pro:Lie-Hopf}. Then $(U(\g),\huaF_{\bullet}U(\g),\huaR)$ is a filtered Rota-Baxter Hopf algebra.
\end{thm}
\begin{proof}
 We need to prove that $\huaR( \huaF_nU(\g)) \subset \huaF_n U(\g)$ for any $n\geq 1$.

Let $x_1,\ldots,x_m\in \g$ and $x_1\ldots x_m\in \huaF_n U(\g)$. If $m=1$, then $x_1\in \huaF_n\g$.  From \eqref{e10} we have that $$\huaR(x_1)=R(x_1)\in \huaF_n \g\subset \huaF_n U(\g).$$

We will use induction on $m$. Consider $x_1\ldots x_m$ with $x_i\in \huaF_{k_i}\g$ and $\sum\limits_i k_i\geq n$. Then by \eqref{e11}
\begin{eqnarray*}
    \huaR(x_1 x_2\ldots x_m)&=&R(x_1)\huaR\big(x_2\ldots x_m)-\huaR([R(x_1),x_2\ldots x_m]\big)
    \\&=&R(x_1)\huaR(x_2\ldots x_n)-\sum\limits_{i=2}^m \huaR\big(x_2\ldots x_{i-1}[R(x_1),x_i]x_{i+1}\ldots x_m\big).
\end{eqnarray*}

By the induction hypothesis and since $R$ preserves the filtration on $\g$, we have $$R(x_1)\huaR(x_2\ldots x_m) \in \huaF_n U(\g).$$ Furthermore, it is obvious that for any $2 \leq i\leq m$, $$\huaR(x_2\ldots x_{i-1}[R(x_1),x_i]x_{i+1}\ldots x_m)\in \huaF_nU(\g).$$ Thus, $\huaR(x_1\ldots x_m)\in \huaF_n U(\mathfrak g)$.
\end{proof}

\begin{defi}
    Let $(\g, \huaF_{\bullet}\g, R)$ be a filtered Rota-Baxter Lie algebra. Then the filtered Rota-Baxter Hopf algebra $(U(\g),\huaF_{\bullet}U(\g),\huaR)$ with the filtration defined by \eqref{g2}-\eqref{g3} and the Rota-Baxter operator $\huaR$ defined by \eqref{e100}-\eqref{e11} is called the  {\bf filtered Rota-Baxter universal enveloping algebra}.
\end{defi}

Let $(\g, \huaF_{\bullet}\g, R)$ be a filtered Rota-Baxter Lie algebra, and $\g_R$ be the corresponding descendent Lie algebra. The
Poincar\'e--Birkhoff--Witt theorem states that universal enveloping algebras $U(\g)$ and $U(\g_R)$ are isomorphic as vector spaces. Moreover, from \cite[Theorem 4]{Goncharov} and \cite[Theorem 3.4]{ELM}, the universal enveloping algebra $U(\g_R)$ of the Lie algebra $\g_R$ is isomorphic to the Hopf algebra $(U(\g),\star,\Delta,\eta,\epsilon, S_{\huaR})$. For simplicity, we may assume that $U(\g_R)=(U(\g),\star,\Delta,\eta,\epsilon, S_{\huaR})$. In what follows, for $a,b\in U(\g)$, by $ab$ we will mean the product in the universal enveloping algebra of the Lie algebra $\g$,  while $a\star b$ means the product in $U(\g_R)$
defined by \eqref{nprodU}.

It is straightforward to check that the product in the descendent Lie algebra $\g_R$ preserves the filtration $\huaF_{\bullet}\g$. In other words, $(\g_R,\huaF_{\bullet}\g)$ is also a filtered Lie algebra. Let $(U(\g_R),\huaF_{\bullet}U(\g_R))$ be the filtered universal enveloping algebra of the filtered Lie algebra $(\g_R,\huaF_{\bullet}\g)$, that is
$$
\huaF_nU(\g_R)=span \left\{x_1\star x_2\star\ldots \star x_k |\ k\geq 1,\
x_i\in \huaF_{n_i}\g,\
\sum_{i=1}^k n_i\ge n\   \right\},\quad\ n\geq 1.
$$
\begin{pro}\label{filt}
For any $n\geq 1$, $\huaF_nU(\g)=\huaF_nU(\g_R)$.
\end{pro}

\begin{proof}
    For any $t\geq 1$, define a subspace $\g^{[t]}=\mathrm{span}\left(\{x_1\ldots x_i|\ x_j\in \g,\ i\leq t\}\right)\subset U(\g)$.
Given $x\in \g$ and $h\in \g^{[t]}$, from \eqref{nprodU} we have
    \begin{equation}\label{stareq}
    x\star h=xh+[R(x),h].
    \end{equation}
      A simple induction shows that $x\star h=xh+y$, where {$y\in \g^{[t]}$}. Moreover, if $x\in \huaF_i\g$, $h\in \huaF_jU(\g)$, then $x\star h$, $xh$ and $y$ are elements from $\huaF_{i+j}U(\g)$. In particular, $\huaF_nU(\g_R)\subset \huaF_nU(\g)$.

   If $x\in \g$, then obviously $x\in\huaF_nU(\g)$ if and only if $x\in \huaF_n(U(\g_R))$.  Assume that for any $t<k$ and $h=y_1\ldots y_t\in \g^{[t]}$, an inclusion $y_1\ldots y_t\in \huaF_nU(\g)$ implies that  $y_1\ldots y_t\in \huaF_nU(\g_R)$.

  Consider $x_1\in \huaF_{i_1}\g,\ldots,x_k\in \huaF_{i_k}\g$ such that $i_1+\ldots +i_k=n$. Using \eqref{stareq}, we have
    $$
    x_1\star x_2\star\ldots \star x_k=x_1\ldots x_k+y,
    $$
where $y\in \huaF_n U(\g)\cap \g^{[k-1]}$. By the assumption, $y\in \huaF_nU(\g_R)$. Since  $x_1\star x_2\star\ldots\star x_k\in \huaF_nU(\g_R)$ (by the definition), we deduce that $x_1,\ldots x_k\in \huaF_nU(\g_R)$.
\end{proof}

Let $m_{\star}: U(\g)\otimes U(\g)\rightarrow U(\g)$ be the map defined by
$$
m_{\star}(a\otimes b)=a\star b.
$$
By Proposition \ref{pro:Lie-Hopf} and Proposition \ref{filt}, $m_{\star}$ is a homomorphism of filtered vector spaces $(U(\g)\otimes U(\g),\huaF_{\bullet} (U(\g)\otimes U(\g))$ and $(U(\g),\huaF_{\bullet}\g)$. Similarly, $S_{\huaR}:U(\g)\rightarrow U(\g)$ is also a homomorphism of filtered vector spaces.

\begin{cor}\label{corfilt}
As vector spaces,
 $\widehat{U}(\g)=\underleftarrow{\lim}U(\g)/\huaF_n U(\g)=\widehat{U}(\g_R)=\underleftarrow{\lim}U(\g_R)/\huaF_n U(\g_R).$
\end{cor}

Let $(\g,\huaF_{\bullet}\g,R)$ be a filtered Rota-Baxter Lie algebra, $(U(\g),\huaF_{\bullet}U(\g),\huaR)$ be the filtered Rota-Baxter universal enveloping algebra of $(\g,\huaF_{\bullet}\g,R)$ given in Theorem \ref{thm:RBLie-RBH}, and $\widehat U(\g)$ be the completion of the filtered universal enveloping algebra $(U(\g),\huaF_{\bullet} U(\g))$.
By $\huaR(\huaF_{n+1} U(\mathfrak g))\subset \huaF_{n+1} U(\mathfrak g)$, there is the induced map $\huaR_n: U(\mathfrak g)/\huaF_{n+1}U(\mathfrak g)\rightarrow U(\mathfrak g)/\huaF_{n+1}U(\mathfrak g)$ given by
\begin{eqnarray}\label{comp-RB}
\huaR_n(x+\huaF_{n+1} U(\mathfrak g))=\huaR(x)+\huaF_{n+1} U(\mathfrak g).
\end{eqnarray}

Obviously, $p_n\circ \huaR_n= \huaR_{n-1}\circ p_n$. Thus, we can define a map
$\widehat \huaR=\varprojlim \huaR_n :\widehat U(\g)\rightarrow \widehat U(\g)$
as
$$
\widehat \huaR(x_0,\ldots,x_n,\ldots)=(\huaR_0(x_0),\ldots,\huaR_n(x_n),\ldots)\in \varprojlim (U(\mathfrak g)/\huaF_{n+1}(U(\mathfrak g)).
$$

Now we are ready to give the main result of the paper.

\begin{thm}\label{thm:main}
    Let $(\g, \huaF_{\bullet}\g, R)$ be a filtered Rota-Baxter Lie algebra. Then the map $\frkR:\widehat{\g}\rightarrow \widehat{\g}$ defined as
\begin{equation}\label{fint}
\frkR(x)=\log(\widehat{\huaR}(\exp(x))),\quad \forall x\in \widehat{\g}
\end{equation}
   is a Rota-Baxter operator on the group $(\widehat{\g},*)$ given in Proposition \ref{pro:group-BCH}.

   In particular, if $(\g, \huaF_{\bullet}\g, R)$ is a complete Rota-Baxter Lie algebra, then   $(\g, *, \frkR)$ with the map $\frkR$ defined by \eqref{fint} for all $x\in \g$ is a Rota-Baxter group, which is considered as the {\bf formal integration} of the complete Rota-Baxter Lie algebra $(\g, \huaF_{\bullet}\g, R)$.

\end{thm}

To prove this theorem, we need the following results, which are interesting on their own.

\begin{pro}\label{pro:RBcompleteH}
$\widehat \huaR$ is a Rota-Baxter operator on the complete Hopf algebra $\widehat U(\g)$.
\end{pro}

\begin{proof} We can rewrite \eqref{RBH} as
$$
\huaR(a)\huaR(b)=\huaR(m_{\star}(a\otimes b))=\huaR(a\star b)
$$
for any $a,b\in U(\g)$. Since $\huaR$, $m_{\star}$, and the product $\mu $ in $U(\g)$ are homomorphisms of the corresponding filtered vector spaces, we obtain that the same equality holds for the completions of these maps, that is
$$
\widehat{\huaR}(a)\widehat{\huaR}(b)=\widehat{\huaR}(a\star b)
$$
for any $a,b\in \widehat{U}(\g)$.
\end{proof}

\begin{pro}\label{pro:RBgrouplike}
$\widehat{\huaR} (\G)\subset \G$, and the restriction $\widehat{\huaR}|_\G:\G\rightarrow \G$ is a Rota-Baxter operator on the group $\G$.
\end{pro}

{\bf The proof of Theorem \ref{thm:main}:} By Proposition \ref{pro:RBcompleteH} and \ref{pro:RBgrouplike}, $\frkR$ is a Rota-Baxter operator on the group $\G$. Since $\log$ and $\exp$ are group isomorphism between $\G$ and $(\widehat\g,*)$, we have
\begin{eqnarray*}
  \frkR(x)*\frkR(y)&=&\log(\widehat{\huaR}(\exp(x)))*\log(\widehat{\huaR}(\exp(y)))\\
  &=&\log\Big(\widehat{\huaR}(\exp(x)) \cdot \widehat{\huaR}(\exp(y))\Big)\\
  &=&\log\Big(\widehat{\huaR}\big(\exp(x)  \cdot \widehat{\huaR}(\exp(x)) \cdot \exp(y) \cdot \widehat{\huaR}(\exp(x))^{-1}\big)\Big)\\
  &=&\log\Big(\widehat{\huaR}\big(\exp(x)  \cdot  \exp(\frkR(x)) \cdot \exp(y) \cdot  (\exp(\frkR(x))^{-1})\big)\Big)\\
  &=&\log\Big(\widehat{\huaR}(\exp(x*   \frkR(x)  *  y  *  (\frkR(x))^{-1}))\Big)\\
  &=&\frkR(x*   \frkR(x)  *  y  *  (\frkR(x))^{-1}),
\end{eqnarray*}
which implies that $\frkR$ is a Rota-Baxter operator on the group $(\widehat\g,*)$. The proof is completed.

 \vspace{2mm}
We can find several components of $\frkR:\g\rightarrow \g$. Let $x=(x_1,x_2\ldots,x_n,\ldots)\in \widehat{\g}$. Then
\begin{eqnarray*}
\frkR(x)&=&\log(\widehat{\huaR}(\exp(x)))\\
&=&\log((\widehat{\huaR}(1,1+x_1,1+x_2+\frac{x_2^2}{2!},\ldots)\\
&=&\log(\huaR(1),\huaR(1)+\huaR(x_1),\huaR(1)+\huaR(x_2)+\huaR(\frac{x_2^2}{2})),\ldots\\
&=&\log((1,1+R(x_1),1+R(x_2)+\frac{1}{2}(R(x_2)R(x_2)-R([R(x_2),x_2]),\ldots ).
\end{eqnarray*}

Thus,
\begin{equation}\label{PRES}
\frkR(x_1,\ldots,x_n,\ldots)=(R(x_1),R(x_2)-\frac{1}{2}R([R(x_2),x_2]),\ldots).
\end{equation}

 From \eqref{PRES} and \eqref{corresp}, we have the following result.

 \begin{pro}\label{nil3}
     Let $\g$ be a nilpotent Lie algebra of nilindex $3$, $\huaF_{\bullet}\g$ be the standard filtration on $\g$, and $R:\g\rightarrow \g$ be a Rota-Baxter operator preserving the filtration. Then the Rota-Baxter operator $\mathfrak R$ on the group $(\g,*)$ is given by
     $$
    \frkR(x)=R(x)-\frac{1}{2}R([R(x),x]).
     $$
 \end{pro}

Note that Proposition \ref{nil3} was independently obtained in \cite{NR}, using a different technique.

\begin{ex}
Consider   the 3-dimensional Heisenberg Lie algebra given in Example \ref{ex:Hei-group}.
Consider an operator $R$:
$$
R(x)=y,\ R(y)=x,\ R(z)=-z.
$$
It is straightforward to check that $R$ is a Rota-Baxter operator of weight 1 on $\mathfrak g$. Obviously, $R$ preserves the filtration, that is, $(\g,\huaF_{\bullet}\g,R)$ is a complete Rota-Baxter Lie algebra.

From Proposition \ref{nil3},  the map $\mathfrak R:\g\rightarrow \g$ defined as
$$
\frkR(\alpha x+\beta y+\gamma z)=\beta x+\alpha y+\frac{1}{2}(-\alpha^2+\beta^2-2\gamma)z
$$
for all $\alpha,\beta,\gamma\in \bk$, is a Rota-Baxter operator on the group $(\g,*)$.

\end{ex}

\begin{ex}\label{gpoly}
Let $(\g,R)$ be a Rota-Baxter Lie algebra over a field $\bk$ of characteristic zero and $\bk[h]h$ be the polynomials in one variable $h$ without a constant term. Consider
$$\g_h=\g[h]h=\g\otimes \bk[h]h=\g h\oplus \g h^2\oplus\ldots\oplus \g h^n\oplus\ldots.
$$
We can continue the map $R$ to a linear map on $\g_h$ (that we will also denote by $R$) by letting $R(xh^k)=R(x)h^k$ for any $x\in \g$ and $k\geq 1$. Define $\huaF_{k}\g_h=\g[h]h^k$ for any $k\geq 1$. Then $(\g_h,\huaF_{\bullet}\g_h,R)$ is a filtered Rota-Baxter Lie algebra, and the completion
of $\g_h$ may be presented as
$$
\widehat{\g_h}=\g[[h]]h=\left\{\sum\limits_{i\geq 1}g_ih^i|\ g_i\in \g\right\},
$$
the formal power series with coefficients in $\g$ without a constant term. By Theorem \ref{thm:main}, we can integrate $R$ to a Rota-Baxter operator $\frkR$ on the group $(\widehat{\g_h},*)$. By \eqref{PRES}, for any $x\in \g$ and $k\geq 1$, we have
$$
\frkR(xh^k)=R(x)h^k-\frac{1}{2}R([R(x),x])h^{2k}+\ldots.
$$

\end{ex}

\begin{pro}
    Let $(\g, \huaF_{\bullet}\g, R)$ be a filtered Rota-Baxter Lie algebra, $(\widehat{\g},*,\frkR)$ be the corresponding Rota-Baxter group given in Theorem \ref{thm:main}, $\phi:\g\rightarrow \g$ be an automorphism of the filtered Lie algebra $(\g,\huaF_{\bullet}\g)$, and $Q=\phi^{-1}R\phi$. Then  $(\g, \huaF_{\bullet}\g, Q)$ is a filtered Rota-Baxter Lie algebra. Moreover, the corresponding  Rota-Baxter operator $\mathfrak Q$   on the group $(\widehat{\g},*)$ is given by $\mathfrak Q=\widehat{\phi}^{-1}\frkR\widehat{\phi}$, where $\widehat{\phi}$ is the natural continuation of the automorphism $\phi$ on the completion $\widehat{\g}$.
\end{pro}

\begin{proof}
Since the automorphism $\phi$ preserves the filtration, $(\g, \huaF_{\bullet}\g, Q)$ is a filtered Rota-Baxter Lie algebra. We can continue $\phi$ from the Lie algebra $\g$ to the universal enveloping algebra $\Phi:U(\g)\rightarrow U(\g)$.  Let $\huaR$ and $\mathcal{Q}$ be continuations of Rota-Baxter operators $R$ and $Q$ into $U(\g)$. From \eqref{e100}-\eqref{e11} we can conclude that $\mathcal{Q}=\Phi^{-1}\huaR\Phi$. Therefore, for any $x\in \widehat \g$, we have
$$
\mathfrak Q(x)=\log( \widehat{\mathcal Q}(\exp(x)))=\log(\widehat{\Phi}^{-1}(\widehat{\huaR}(\widehat{\Phi}(\exp{x}))))=\log(\exp((\widehat{\phi}^{-1}{\frkR}\widehat{\phi})(x))).
$$
Thus, $\mathfrak Q=\widehat{\phi}^{-1}\frkR\widehat{\phi}$. 
\end{proof}

Let $(\g, \huaF_{\bullet}\g, R)$ be a filtered Rota-Baxter Lie algebra and $x\in \g$. Let $L_R(x)$ be the minimal $R$-invariant subalgebra in $\g$ containing $x$. Let $R_1$ be the restriction of $R$ on $L_R(x)$, then $(L_R(x),\huaF_{\bullet}\g\cap L_R(x),R_1)$ is a filtered Rota-Baxter Lie algebra. The universal enveloping algebra $U(L_R(x))$ of the filtered Lie algebra $(L_R(x),\huaF_{\bullet}\g\cap L_R(x))$ may be naturally embedded (as a filtered Hopf algebra) into $(U(\g),\huaF_{\bullet}U(\g))$.

Let $\pi:\g\lon\varprojlim \g/\huaF_{n+1}\g$ be the natural homomorphism from $\g$ to $\widehat{\g}$.
We have $\ker\pi=
\bigcap\limits_{n=0}^{\infty}\huaF_{n+1}\g$. In general, there is an embedding of the quotient algebra $\g/\ker\pi$ into $\widehat{\g}$. In particular, the map $\pi:\g\lon\varprojlim \g/\huaF_{n+1}\g$ is an embedding if and only if $\bigcap\limits_{n=0}^{\infty}
\huaF_{n+1}\g=\{0\}$.

\begin{pro}
Let $(\g, \huaF_{\bullet}\g, R)$ be a filtered Rota-Baxter Lie algebra, $\frkR$ be the Rota-Baxter operator on the group $(\widehat {\g},*)$ defined by \eqref{fint}, and $x\in \g$. Then
\begin{itemize}
\item[{\rm(i)}] $\frkR(y)\in \widehat{L_R(x)}$ for any $y\in L_R(x)$.
\item[{\rm(ii)}] If $[R(x),x]=0$, then $\frkR(\pi(x))=\pi(R(x))\in \pi(\g)\subset \widehat{\g}$.
\item[{\rm(iii)}] If $R(x)=0$, then $\frkR(\pi(x))=0$.
\end{itemize}
\end{pro}
\begin{proof}
(i) By \eqref{e100}-\eqref{e11}, it follows that $U(L_R(x))$ is an $\huaR$-invariant Hopf subalgebra in $U(\g)$, and the restriction of the Rota-Baxter operator $\huaR|_{U(L_R(x))}$ coincides with $\huaR_1$, the continuation of the Rota-Baxter operator $R_1$ to $U(L_R(x))$. Therefore, for any $y\in L_R(x)$,
$$\widehat{\huaR}(\exp(y))=\widehat{\huaR_1}(\exp(y))\in \widehat{U}(L_R(x))\subset \widehat{U}(\g)$$
and
$$\frkR(y)=\log(\widehat{\huaR}(\exp(y)))=\log(\widehat{\huaR_1}(\exp(y)))\in \widehat{L_R(x)}$$
by Theorem \ref{thm:main}.

(ii) If $[R(x),x]=0$, then using simple induction and \eqref{e11}, we obtain that $\huaR(x^n)=R(x)^n$. Thus, $\widehat{\huaR}(\exp(\pi(x)))=\exp(\pi(R(x)))$ and  $$\frkR(\pi(x))=\log(\widehat{\huaR}(\exp(\pi(x))))=\log(\exp(\pi(R(x))))=\pi(R(x))\in \pi(\g).$$

(iii) Follows from (ii).
\end{proof}

Let $\g=(\g,\huaF_{\bullet}\g,R)$ be a filtered Rota-Baxter Lie algebra. By Corollary \ref{corfilt}, the completion of the filtered universal enveloping algebra of $(\g,\huaF_{\bullet}\g)$ can also be considered as the completion of the descendent Lie algebra $\g_R=(\g,[\cdot,\cdot]_R,\huaF_{\bullet}\g)$. Note that  $\widehat{\g}=\widehat{\g}_R$ as vector spaces. By Proposition \ref{grexp}, we may define two ``exponents'' $\exp:\widehat{\g}\rightarrow \widehat{U}(\g)$ and $\exp_{\star}:\widehat{\g}\rightarrow \G\subset \widehat{U}(\g)$ given by
$$
\exp(x)=1+x+\frac{x^2}{2}+\ldots+\frac{x^n}{n!}+\ldots
$$
and
$$
\exp_{\star}(x)=1+x+\frac{x{\star}x}{2}+\ldots+\frac{x^{n,\star}}{n!}+\ldots,
$$
where $x\in \widehat{\g}$ and $x^{n,\star}=x{\star}x{\star}x\ldots{\star}x$ ($n$ times). Similarly, we have two ``logarithms'' $\log=\exp^{-1}:\G\rightarrow \widehat{\g}$ and $\log_{\star}=\exp_{\star}^{-1}:\G\rightarrow \widehat{\g}$.

In particular, for any $x\in \widehat{\g}$, there is a unique element $t\in \widehat{\g}$ such that $\exp(x)=\exp_{\star}(t)$. This naturally lead us to the following notion
of the {\bf post-Lie Magnus expansion} \mcite{CP,CEO,EMQ,EMM,MQS}  that is given by
\begin{eqnarray}\mlabel{post-Lie Magnus expansion}
 \Omega:\widehat{\g}\lon \widehat{\g}, \quad \Omega(x):=\log_{\star}(\exp(x)) \tforall x\in \widehat{\g}.
\end{eqnarray}

We will need the following result that was originally obtained in \cite[Theorem 17]{EMM}.
\begin{pro}\label{Magnus}
    For any $x\in \widehat{\g}$, we have
\begin{eqnarray}
\Omega(x)=\Omega_1(x)+\Omega_2(x)+\ldots+\Omega_n(x)+\ldots
\end{eqnarray}
where
\begin{eqnarray}
\nonumber \Omega_1(x)&=&x,\\
\Omega_n(x)&=&\left(\frac{1}{n!}x^n-\sum\limits_{k=2}^{n}\sum\limits_{i_1+\ldots+i_k=n\atop i_1,\dots,i_k\ge 1}\frac{1}{k!}\Omega_{i_1}(x)\star \Omega_{i_2}(x)\star \ldots\star \Omega_{i_k}(x)\right)\in \huaF_n\widehat{\g},\quad n\geq 2.\label{form}
\end{eqnarray}
Here  $\Omega_n(x)$ is the degree $n$ part of the $\Omega(x)$ and $\star $ is the product in ${U}(\g)$ defined by \eqref{nprodU}.
\end{pro}

 Let $\widehat{R}:\widehat{\g}\rightarrow\widehat{\g}$ be the completion of $R$ defined by
\begin{eqnarray}\label{completion-RBA}
\widehat{R}(x_1,\ldots,x_n,\ldots)=(R_1(x_1),\ldots,R_n(x_n),\ldots),
    \end{eqnarray}
here $(x_1,\ldots,x_n,\ldots)\in \widehat{\g}=\underleftarrow{\lim}\ \g/\huaF_{n+1} (\g)$ and $R_n:\g/\huaF_{n+1} (\g)\lon \g/\huaF_{n+1} (\g)$ the induced Rota-Baxter operator of $R$. Then $(\widehat{\g}, \huaF_{\bullet}\widehat{\g} , \widehat{R})$ is a complete Rota-Baxter Lie algebra.

\begin{thm}\label{thm:formula}
Let $(\g, \huaF_{\bullet}\g, R)$ be a filtered Rota-Baxter Lie algebra, $\frkR$ be the Rota-Baxter operator on  the group $(\widehat {\g},*)$ defined by \eqref{fint}.  Then $\frkR=\widehat{R}\circ \Omega$.
\end{thm}

\begin{proof}

Since $(\widehat{U}(\g),\cdot,\Delta, \eta, \epsilon, S,\huaF_{\bullet}U(\g),\widehat{\huaR})$ is a  complete Rota-Baxter Hopf algebra, we deduce that
\begin{eqnarray}\label{RBH-homo}
\widehat{\huaR}(x\star y)=\widehat{R}(x)\widehat{R}(y),\,\,\forall x,y\in \widehat{\g}.
\end{eqnarray}
As a corollary, we get  $\widehat{\huaR}(\exp_{\star}(x))=\exp(\widehat{R}(x))$ for any $x\in \widehat{\g}$. Therefore, we have
\begin{eqnarray*}
\exp(\frkR(x))=\widehat{\huaR}(\exp(x))&\stackrel{\eqref{post-Lie Magnus expansion}}{=}&\widehat{\huaR}\Big(\exp_\star \big(\Omega(x)\big)\Big)\stackrel{\eqref{RBH-homo}}{=} \exp\Big(\widehat{R}\big(\Omega(x)\big)\Big).
\end{eqnarray*}
Then, we deduce
$
\frkR(x)=\log(\widehat{\huaR}(\exp(x)))=\log\left(\exp\Big(\widehat{R}\big(\Omega(x)\big)\Big)\right)=\widehat{R}\left(\Omega(x)\right).
$
\end{proof}

\begin{rmk}
The post-Lie Magnus expansion  occupies  an important position in the  formal integration theory of post-Lie algebras \cite[Theorem 5.22]{BGST}.
It is very surprising that post-Lie Magnus expansion also naturally emerges in  the formal integration theory of Rota-Baxter Lie algebras. Please see \cite{EMQ} for the latest developments of the Magnus expansion and its applications.
\end{rmk}

\begin{ex}
We can compute $\Omega_2(x)$ and $\Omega_3(x)$. Note that $\Omega_1(x)\star\Omega_1(x)=x\star x=x^2+[R(x),x]$. From \eqref{form}, we have
$$
\Omega_2(x)=\frac{1}{2}x^2-\frac{1}{2}(\Omega_1(x)\star\Omega_1(x))=\frac{1}{2}(x^2-x^2-[R(x),x])=-\frac{1}{2}[R(x),x].
$$
Note that this agrees with \eqref{PRES}. In order to compute $\Omega_3(x)$, we need the following  equalities:
\begin{eqnarray*}
\Omega_1(x)\star\Omega_1(x)\star\Omega_1(x)&=&x\star (x^2+[R(x),x])=x^3+x[R(x),x]+[R(x),x^2]+[R(x),[R(x),x]],\\
\Omega_1(x)\star\Omega_2(x)&=&-\frac{1}{2}(x[R(x),x]+[R(x),[R(x),x]]),\\
\Omega_2(x)\star \Omega_1(x)&=&-\frac{1}{2}([R(x),x]x+[R([R(x),x]),x]).
\end{eqnarray*}
Then
\begin{eqnarray*}
\Omega_3(x)&= & \frac{1}{6}x^3-\frac{1}{2}\big(\Omega_1(x)\star\Omega_2(x)+\Omega_2(x)\star \Omega_1(x)\big)-\frac{1}{6}\Omega_1(x)\star\Omega_1(x)\star\Omega_1(x)\\
 & =& \frac{1}{6}x^3+\frac{1}{4}([R(x),x^2]+[R(x),[R(x),x]]+[R([R(x),x]),x])\\
 && -\frac{1}{6}(x^3+x[R(x),x]+[R(x),x^2]+[R(x),[R(x),x]])\\
 & =& \frac{1}{12}[R(x),x^2]-\frac{1}{6}x[R(x),x]+\frac{1}{12}[R(x),[R(x),x]]+\frac{1}{4}[R([R(x),x]),x]\\
 & = &\frac{1}{12}[[R(x),x],x]+\frac{1}{12}[R(x),[R(x),x]]+\frac{1}{4}[R([R(x),x]),x].
\end{eqnarray*}
That is,
\begin{equation}
    \Omega_3(x)=\frac{1}{12}[[R(x),x],x]+\frac{1}{12}[R(x),[R(x),x]]+\frac{1}{4}[R([R(x),x]),x].
\end{equation}
\end{ex}

\begin{cor}
Let $\g$ be a nilpotent Lie algebra of nilindex $4$ and  $R:\g\rightarrow \g$ be a Rota-Baxter operator on $\g$ satisfying  $R(\g^i)\subset \g^i$ for $i=2,3$. Then the map $\frkR$ defined as
$$
\frkR(x)=R(x)-\frac{1}{2}R([R(x),x])+\frac{1}{12}R([[R(x),x],x]+[R(x),[R(x),x]])+\frac{1}{4}([R([R(x),x]),x])
$$
for any $x\in \g$, is a Rota-Baxter operator on  the group $(\g,*)$.
\end{cor}

\section{From filtered Rota-Baxter groups to graded Rota-Baxter Lie rings}\label{sec:diffRB}

In this section, we show that one can obtain a graded Rota-Baxter Lie ring from a filtered Rota-Baxter group.

Let $G$ be a group. Given $x,y\in G$, denote by $(x,y)$ the commutator of elements $x,y$:
$$
(x,y)=xyx^{-1}y^{-1}.
$$
By $x^y$ we will mean the conjugation: $x^y=yxy^{-1}$. Let $H$  and $K$ be two subgroups of $G$. We use the notation $(H,K)$  for the subgroup of $G$ generated by  the commutators $(x,y)$ for $x\in H,~y\in K$.

\begin{defi}
    A filtered group is a pair $(G,\huaF_{\bullet}G)$, where $G$ is a group and $\huaF_{\bullet}G$ is a descending filtration (the integral filtration in the terminology of \cite{Ser}) of the group $G$ such that $G=\huaF_1G\supset\huaF_2G\supset\cdots\supset\huaF_n G\supset\cdots$, and
    \begin{equation}\label{ugr}
        (\huaF_nG,\huaF_mG)\subset \huaF_{m+n}G
    \end{equation}
    for all $n,m\geq 1$.
\end{defi}


The following properties will be frequently used in the sequel.

\begin{pro}\label{Ser}\cite{Ser} Let $(G,\huaF_{\bullet}G)$ be a filtered group. For any $n\geq 1$ we have
    \begin{itemize}
\item[{\rm(i)}] $\huaF_n G$ is a normal subgroup in $G$;
\item[{\rm(ii)}] $\huaF_nG/\huaF_{n+1}G$ is an abelian group;
\item[{\rm(iii)}] For any $x\in \huaF_n G$ and $y\in G$: $x^y\in \huaF_{n}G$.
\end{itemize}
\end{pro}

\begin{ex}\label{exgrfilt}
   Let $G$ be an arbitrary group and $G=G^1>G^2>\ldots >G^n>\ldots$ be the descending central series of the group $G$ defined as $G^1=G$ and $G^n=(G,G^{n-1})$ for $n>1$. Then the sequence $\{\huaF_nG=G^n\}$ is a filtration on the group $G$. Moreover, this is the minimal filtration in the sense that if $\huaF_{\bullet}'G$ is another filtration on $G$, then $G^n\subset \huaF'_nG$ for all $n\geq 1$ (see \cite{Ser}).
\end{ex}

\begin{ex}\label{filthg}
Let $(H,\huaF_{\bullet}H)$ be a filtered Hopf algebra. Then the group $G$ of group-like elements of $H$ has an induced filtration:
$$
\huaF_n G=(1+\huaF_nH)\cap G
$$
 Indeed,
$\epsilon(g)=1$ for every $g\in G$, so
$x_g=g-1 \in \ker(\epsilon)$, i.e., $g=1+x_g$.
 Suppose that $g$ and $h$   are two group-like elements such that $g\in \huaF_nG$ and $h\in \huaF_mG$. This means that $x_g\in\huaF_n H$ and $x_h\in \huaF_m H$. Then the commutator $(g,h)$ is again a group-like element. Let $(g,h)=1+t$, where $t\in \ker(\epsilon)$. Then
\begin{eqnarray*}
 gh&=&(1+t)hg,\\
 (1+x_g)(1+x_h)&=&(1+t)(1+x_h)(1+x_g),\\
 1+x_g+x_h+x_gx_h&=&1+x_h+x_g+x_hx_g+t(1+x_h+x_g+x_hx_g),\\
 t&=&[x_g,x_h]-t(x_h+x_g+x_hx_g).
\end{eqnarray*}
     Therefore, $t=[x_g,x_h]+ta$, where $a\in\ker(\epsilon)=\huaF_1H$. But then
 \begin{eqnarray*}
     t&=&[x_g,x_h]+ta=[x_g,x_h]+[x_g,x_h]a+ta^2=\ldots
     \\&=&[x_g,x_h]+[x_g,x_h]a+\ldots+[x_g,x_h]a^{n+m-1}+ta^{m+n}\in\huaF_{m+n}H.
\end{eqnarray*}
 That means that $(G,\huaF_{\bullet}G)$ is a filtered group.
\end{ex}

\begin{rmk}\label{rmkfilt}
    If $(\g,\huaF_{\bullet}\g)$ is a filtered Lie algebra, and $G$ is the group of group-like elements of the complete Hopf algebra $\widehat U(\g)$, then the filtration from Example \ref{filthg} may be presented as follows:
    $$\huaF_nG=\{\exp(x)|\ x\in\huaF_n \widehat{\g}\},\quad n\geq 1.$$
\end{rmk}

\begin{ex}
    Let $(\g,\huaF_{\bullet}\g)$ be a filtered Lie algebra. Then the triple $(\widehat{\g},*,\huaF_{\bullet}\widehat{\g})$, where $x*y=\BCH(x,y)$, is a filtered group.
\end{ex}

  Let $(G,\huaF_{\bullet}G)$ be a filtered group.
For any $n\geq 1$, denote by
$$\gr_nG=\huaF_nG/\huaF_{n+1}G
$$
which is an abelian group and the group structure ``+'' is defined by
$$
\overline{x}+\overline{y}=\overline{xy},\quad \forall x,y\in \huaF_nG.
$$
 Denote by
$$
\gr G=\gr_1G\oplus \gr_2G\oplus\ldots\oplus \gr_nG\oplus\ldots,
$$
 the direct sum of abelian groups.

Define a graded product on the group $\gr G$ as follows: if $a=\overline{x}\in \gr_nG$ and $b=\overline{y}\in \gr_mG$, then the product $[a,b]$ is defined to be an element in $\gr_{m+n}G$ given by
\begin{equation}\label{prod}
    [a,b]=\overline{(x,y)}\in \gr_{n+m}G=\huaF_{n+m}G/\huaF_{n+m+1}G.
\end{equation}

\begin{pro}\cite{Ser}
 Let $(G,\huaF_{\bullet}G)$ be a filtered group. Then $(\gr G,+,[\cdot,\cdot])$ is a graded Lie ring.
\end{pro}


\begin{defi}
    A \textbf{filtered Rota-Baxter group} is triple $(G,\huaF_{\bullet}G,\frkR)$, where $(G,\huaF_{\bullet}G)$ is a filtered group, $\frkR$ is a Rota-Baxter operator on the group $G$ and
    $$
    \frkR(\huaF_nG)\subset \huaF_n G,\quad \forall n\geq 1.
    $$

\end{defi}

\begin{ex}
Let $(H,\huaF_{\bullet}H,\huaR)$ be a filtered Rota-Baxter Hopf algebra, $(G,\huaF_{\bullet}G)$ be the group of group-like elements with the filtration given by Example \ref{filthg} and $\frkR$ be the induced Rota-Baxter operator on the group $G$ defined by   $\frkR=\huaR|_G:G\rightarrow G$. Then $(G,\huaF_{\bullet}G,\frkR)$ is a filtered Rota-Baxter group.
Indeed, for any $g\in \huaF_nG$ such that $g=1+x_g$ for some $x_g\in \huaF_nH$, we have
\begin{eqnarray*}
\frkR(g)=\frkR(1+x_g)=\frkR(1)+\frkR(x_g)=1+\frkR(x_g)\in (1+\huaF_nH)\cap G=\huaF_nG.
\end{eqnarray*}
\end{ex}

\begin{pro}
    Let $(G,\huaF_{\bullet}G,\frkR)$ be a filtered Rota-Baxter group. Then for any $n\geq 1$, the map $\frkR_n:\gr_nG\rightarrow \gr_nG$ defined by
\begin{eqnarray*}
    \frkR_n(\overline{x})=\overline{\frkR(x)},\quad \forall \overline{x}\in \gr_nG
\end{eqnarray*}
    is a well-defined homomorphism of $\gr_nG$.
\end{pro}

\begin{proof}
If $\overline{x}=\overline{x'}$ for $x,x'\in \huaF_nG$, then we have $x'=xy$ for $y\in \huaF_{n+1}G$. Since $\frkR$ is a Rota-Baxter operator, we have
$$
\frkR(x')=\frkR(xy)=\frkR(x\frkR(x)\frkR(x)^{-1}y\frkR(x)\frkR(x)^{-1})=\frkR(x)\frkR(y^{\frkR(x)^{-1}}).
$$
Then by Proposition \ref{Ser}, $y^{\frkR(x)^{-1}}\in \huaF_{n+1}G$. Therefore, $ \overline{\frkR(x')}=\overline{\frkR(x)\frkR(y^{\frkR(x)^{-1}})}=\overline{\frkR(x)}$ and $\frkR_n$ is well-defined.

Let $\overline{x},\overline{x'}\in gr_nG$. Then
$$
\frkR_n(\overline{x}+\overline{x'})=\frkR_n(\overline{xx'})=\overline{\frkR(xx')}=\overline{\frkR(x)\frkR\left({x'}^{\frkR(x)^{-1}}\right)}
=\frkR_n(\overline{x})+\frkR_n\left(\overline{{x'}^{\frkR(x)^{-1}}}\right)=\frkR_n(\overline{x})+\frkR_n(\overline{x'}),
$$
which finishes the proof.
\end{proof}

    Define a map $\gr\frkR:\gr G\rightarrow \gr G$ by
  \begin{equation}\label{defRBGr}
    \gr\frkR(a_1+a_2+\ldots +a_k)=\frkR_{i_1}(a_1)+\frkR_{i_2}(a_2)+\ldots+\frkR_{i_k}(a_k),
\end{equation}
    where $a_j\in \gr_{i_j}G$, $j=1,\ldots,k$.

\begin{thm}\label{thm:grring} Let $(G,\huaF_{\bullet}G,\frkR)$ be a filtered Rota-Baxter group. Then $(\gr G,\gr\frkR)$ is a graded Rota-Baxter Lie ring.
\end{thm}

\begin{proof}
For $x,y\in G$ by the fact that $\frkR$ is a homomorphism from the group $(G,\star)$ to $G$, we have
\begin{eqnarray*}
   && (\frkR(x),\frkR(y))\\&=  & \frkR(x)\frkR(y)(\frkR(y)\frkR(x))^{-1}\\
     & =&\frkR(x\star y)(\frkR(y\star x))^{-1}\\
     &=&\frkR(x\star y)\frkR(\frkR(y\star x)^{-1}(y\star x)^{-1}\frkR(y\star x))\\
    & =&\frkR(x\star y)\frkR\big(\frkR(x)^{-1}\frkR(y)^{-1}(y\star x)^{-1}\frkR(y)\frkR(x)\big)\\
    & =&\frkR\Big((x\star y)\frkR(x\star y)\frkR(x)^{-1}\frkR(y)^{-1}(y\star x)^{-1}\frkR(y)\frkR(x)\frkR(x\star y)^{-1}\Big)\\
   & =&\frkR\Big([x\frkR(x)y\frkR(x)^{-1}][\frkR(x)\frkR(y)]\frkR(x)^{-1}\frkR(y)^{-1}[\frkR(y)x^{-1}\frkR(y)^{-1}y^{-1}]\frkR(y)\frkR(x)[\frkR(y)^{-1}\frkR(x)^{-1}]\Big)\\
    & =&\frkR\Big(x\frkR(x)y\frkR(y)\frkR(x)^{-1}x^{-1}\frkR(y)^{-1}y^{-1}(\frkR(y),\frkR(x))\Big)\\
    & =&\frkR\Big(x(\frkR(x),y)y\frkR(x)\frkR(y)\frkR(x)^{-1}x^{-1}\frkR(y)^{-1}y^{-1}(\frkR(y),\frkR(x))\Big)\\
    & =&\frkR\Big(x(\frkR(x),y)y(\frkR(x),\frkR(y))\frkR(y)x^{-1}\frkR(y)^{-1}y^{-1}(\frkR(y),\frkR(x))\Big)\\
    & =&\frkR\Big(xyx^{-1}y^{-1}(\frkR(x),y)^{yxy^{-1}}yxy^{-1}y(\frkR(x),\frkR(y))\frkR(y)x^{-1}\frkR(y)^{-1}y^{-1}(\frkR(y),\frkR(x))\Big)\\
    & =&\frkR\Big((x,y)(\frkR(x),y)^{yxy^{-1}}(\frkR(x),\frkR(y))^{yx}(x,\frkR(y))^y(\frkR(y),\frkR(x))\Big).
\end{eqnarray*}

Then for  $a=\overline{x}\in \gr_nG$ and $b=\overline{y}\in \gr_mG$, where $x\in \huaF_nG$ and $y\in \huaF_m G$, we have
\begin{eqnarray*}
    &&[\gr\frkR(a),\gr\frkR(b)]\\& = {} & [\frkR_n(a),\frkR_m(b)]\\
     & =&(\overline{\frkR(x)},\overline{\frkR(y)})=\overline{(\frkR(x),\frkR(y))}\\
    & =&\overline{\frkR\Big((x,y)(\frkR(x),y)^{yxy^{-1}}(\frkR(x),\frkR(y))^{yx}(x,\frkR(y))^y(\frkR(y),\frkR(x))\Big)}\\
    & =&\frkR_{n+m}\left(\overline{(x,y)}+\overline{(\frkR(x),y)}+\overline{(\frkR(x),\frkR(y))}+\overline{(x,\frkR(y))}+\overline{(\frkR(y),\frkR(x))}\right)\\
    & =&\gr\frkR([a,b]+[\gr\frkR(a),b]+[\gr\frkR(a),R(b)]+[a,\gr\frkR(b)]+[\gr\frkR(b),\gr\frkR(a)])\\
    & =&\gr\frkR([\gr\frkR(a),b]+[a,\gr\frkR(b)]+[a,b]).
\end{eqnarray*}
    That is, $\gr\frkR$ is a Rota-Baxter operator on the Lie ring $\gr G$. It is obvious that $\gr\frkR(\gr_n G)\subset \gr_nG$ for any $n\geq 1$. The proof is finished.
\end{proof}

\begin{rmk}
    If $(G,\huaF_{\bullet}G,R)$ is a Rota-Baxter group with the standard filtration, then the results of Theorem  \ref{thm:grring} were partially obtained in \cite{BarGub}.
\end{rmk}

Note that currently we only obtain graded Rota-Baxter Lie rings from filtered Rota-Baxter groups. At the end of this section, we discuss when can we obtain graded Rota-Baxter $\mathbb Q$-Lie algebras.

Recall that
    a group $G$ is called uniquely divisible if the map $x\mapsto x^m$ is bijective for any $m\geq 1$.
It is straightforward to check that the Lie ring $\gr G$ is a $\mathbb Q$-Lie algebra if and only if $\gr_nG$ is a uniquely divisible abelian group for any $n\geq 1$. In this case, $\gr G$ is a $\mathbb Q$-Lie algebra  with the action given by
\begin{equation}\label{act}
m\cdot\overline{x}=\overline{x^m},
\end{equation}
for all $m\in \mathbb Q$ and $\overline{x}\in \gr_iG$, $i\geq 1$.

Since $\gr\frkR$ is an additive map, we have the following result.
\begin{pro}\label{pro:qcondition}
Let $(G,\huaF_{\bullet}G,\frkR)$ be a filtered Rota-Baxter group. Suppose that $\gr G$ has a structure of a $\mathbb Q$-Lie algebra. Then  the  map $\gr\frkR:\gr G\rightarrow \gr G$ defined by \eqref{defRBGr} is a $\mathbb Q$-linear map and $(\gr G,\gr\frkR)$ is a graded Rota-Baxter $\mathbb Q$-Lie algebra.
\end{pro}

\begin{cor}
 Let $(\g,[\cdot,\cdot]_\g,\huaF_{\bullet}\g,R)$ be a filtered Rota-Baxter Lie algebra over $\mathbb Q$, and $(\widehat{\g},*, \huaF_{\bullet}\widehat{\g},\frkR)$ be the corresponding filtered Rota-Baxter group. Then $(\gr(\widehat{\g},*),\gr\frkR)$ is a graded Rota-Baxter $\mathbb Q$-Lie algebra.
\end{cor}
\begin{proof}
 It is obvious that  $\gr(\widehat{\g},*)$ is a Lie algebra over $\mathbb Q$ with the action defined by
$$
m\cdot \overline{x}=\overline{mx}
$$
for all $m\in\mathbb Q$ and $\overline{x}\in \gr_i(\widehat{\g},*)$, $i\geq 1$. Then by Proposition \ref{pro:qcondition}, $(\gr(\widehat{\g},*),\gr\frkR)$ is a graded Rota-Baxter $\mathbb Q$-Lie algebra.
\end{proof}

Let  $(\g,[\cdot,\cdot]_\g,\huaF_{\bullet}\g,R)$ be a filtered Rota-Baxter Lie algebra over $\mathbb Q$. Denote by $\gr(\widehat{\g},[\cdot,\cdot]_{\widehat{\g}})=\bigoplus\limits_{n\geq 1}\huaF_n\widehat{\g}/\huaF_{n+1}\widehat{\g}$ the induced graded Lie algebra of the completion $(\widehat{\g},[\cdot,\cdot]_{\widehat{\g}}).$ Denote by $\gr\widehat{R}$ the induced Rota-Baxter operator on $\gr(\widehat{\g},[\cdot,\cdot]_{\widehat{\g}})$, that is, $$\gr\widehat{R}(x+\huaF_{n+1}\widehat{\g})=\widehat{R}(x)+\huaF_{n+1}\widehat{\g},\quad \forall x\in \huaF_n\widehat{\g}.$$

\begin{cor}
  Let  $(\g,[\cdot,\cdot]_\g,\huaF_{\bullet}\g,R)$ be a filtered Rota-Baxter Lie algebra over $\mathbb Q$. Then graded Rota-Baxter $\mathbb Q$-Lie algebras $(\gr(\widehat{\g},*),\gr\frkR)$ and  $(\gr(\widehat{\g},[\cdot,\cdot]_{\widehat{\g}}),\gr\widehat{R})$ are isomorphic.
\end{cor}
\begin{proof}
  It follows from  \cite[Proposition 2.8]{Quillen} that graded Lie algebras $\gr(\widehat{\g},*)$ and $\gr(\widehat{\g},[\cdot,\cdot]_{\widehat{\g}})$ are isomorphic. In fact, the isomorphism is induced by maps $\varphi_n:\gr_n(\widehat{\g},*)\rightarrow \gr_n(\widehat{\g},[\cdot,\cdot]_{\widehat{\g}})$ defined as
$$
\varphi_n(x*\huaF_{n+1}\widehat{\g})=x+\huaF_{n+1}\widehat{\g}\in \huaF_n\widehat{\g}/\huaF_{n+1}\widehat{\g},\quad \forall x\in \huaF_n \widehat{\g}.
$$
From Theorem \ref{thm:formula}, for any $x\in \huaF_n\widehat{\g}$ we have that $\frkR(x)=\widehat{R}(x)+\huaF_{2n}{\widehat{\g}}$. Then for any $x\in \huaF_n \widehat{\g}$,
\begin{eqnarray*}
\varphi_n(\gr\frkR(x*\huaF_{n+1}\widehat{\g}))&=&\varphi_n(\frkR(x)*\huaF_{n+1}\widehat{\g})=\varphi_n((R(x)+\huaF_{2n}\widehat{\g})*\huaF_{n+1}\widehat{\g})
\\
&=& \varphi_n(\widehat{R}(x)*\huaF_{n+1}\widehat{\g})=\widehat{R}(x)+\huaF_{n+1}\widehat{\g}=\gr\widehat{R}(x+\huaF_{n+1}\widehat{\g})\\
&=&\gr\widehat{R}(\varphi_n(x*\huaF_{n+1}\widehat{\g})).
\end{eqnarray*}
Therefore, $\varphi_n\circ \gr\frkR=\gr\widehat{R}\circ \varphi_n$ and maps $\{\varphi_n\}$ induce an isomorphism of  graded Rota-Baxter Lie algebras.
\end{proof}


\begin{rmk}
A nilpotent Lie algebra $(\g,[\cdot,\cdot]_\g)$ is called {\bf naturally graded} if $\g$ is isomorphic to $\gr(\g,[\cdot,\cdot]_\g)$ (that is constructed with respect to the standard filtration). It is straightforward to check that the Heisenberg Lie algebra from Example \ref{ex:Hei-group} is naturally graded.
Let $\g$ be a naturally graded nilpotent Lie algebra over $\mathbb Q$, $\huaF_{\bullet}\g$ be the standard filtration on $\g$ and $(\g,*,\huaF_{\bullet}\g)$ be the formal integration of $(\g,\huaF_{\bullet}\g)$. If $(\g,*,\huaF_{\bullet}\g,\frkR)$ is a filtered Rota-Baxter group, then $(\gr(\g,*),\gr\frkR)$, where $\gr\frkR$ is the map defined by \eqref{defRBGr}, is a filtered Rota-Baxter Lie algebra.
\end{rmk}

\begin{ex}
Let $\g$ be a Lie algebra over a field $\mathbb Q$ and  $\g_h=\g[h]h$ be the Lie algebra from Example \ref{gpoly} with the filtration $\huaF_n\g_h=\g[h]h^n$. Then the completion of $\g_h$ is equal to $\g[[h]]h$ with the filtration $\huaF_n\widehat{\g_h}=\g[[h]]h^n$. The graded algebra of $(\widehat{\g_h},\huaF_{\bullet}\widehat{\g_h})$ is isomorphic to $\g[h]h=\g_h$. Therefore, any structure of a filtered Rota-Baxter group $(\g[[h]]h,\huaF_\bullet\g[[h]]h,\frkR)$ gives rise to a structure of a filtered Rota-Baxter Lie algebra on $\g_h$.
\end{ex}

   We finish the paper with an example to illustrate that even though sometimes  it is possible to endow the graded Lie ring $\gr G$ with a structure of an algebra over the field $\bk$ of characteristic zero, but the induced Rota-Baxter operator is not necessarily $\bk$-linear.

\begin{ex}
Let $G=(\mathbb C,+)$ be the additive group of the field of complex numbers with the minimal filtration $\huaF_1\mathbb C=\mathbb C$, $\huaF_n\mathbb C=\mathbb C_n=0$ for any $n\geq 2$ (see Example \ref{exgrfilt}). Since $\mathbb C=\mathbb R\oplus i\mathbb R$ (the direct sum of subgroups), a map $\frkR:G\rightarrow G$ defined as $\frkR(x+iy)=-iy$ for any $x,y\in \mathbb R$ is a Rota-Baxter operator on the group $G$ \cite{GLS}. Thus, $(G,\huaF_{\bullet}G,\frkR)$ is a filtered Rota-Baxter group. The corresponding graded Lie ring is  $\gr G=\mathbb C$ with ordinary addition and trivial multiplication, and the induced map $\gr\frkR$ is equal to $\frkR$, which is not $\mathbb C$-linear since $\gr\frkR(ix)=-ix$ while $i\gr\frkR(x)=0$ for any $x\in \mathbb R$.
\end{ex}

\noindent
{\bf Declaration of interests. } The authors have no conflicts of interest to disclose.

\noindent
{\bf Data availability. } Data sharing is not applicable to this article as no new data were created or analyzed in this study.

\end{document}